\documentclass[runningheads]{llncs}
\usepackage{amsfonts}
\usepackage{amsmath}
\usepackage{amssymb}
\usepackage{calc}
\usepackage{xcolor}
\usepackage{graphicx}
\usepackage{hyperref}
\usepackage{import}
\usepackage{latexsym}
\usepackage{listings}
\usepackage{makeidx}
\usepackage{multicol}
\usepackage{url}
\usepackage{verbatim}
\usepackage{todo}

\title{Clairvoyant State Machine Replication}

\newcommand{\var}[1]{\lstinline+#1+}

\newcommand{\kPropose}{\var{Propose}}
\newcommand{\kConfirm}{\var{Confirm}}

\newcommand{\kProposeAck}{\var{ProposeAck}}

\newcommand{\kResolveAck}{\var{ResolveAck}}

\title{Clairvoyant State Machine Replication}

\author{Rida Bazzi\inst{1} \and Maurice Herlihy\inst{2}}

\institute{Arizona State University, Tempe, AZ, USA \email{bazzi@asu.edu} \and
Brown University, Providence, RI, USA \email{mph@cs.brown.edu}}

\authorrunning{R.\,A. Bazzi and M.\, Herlihy}

\begin{document}

\maketitle

\begin{abstract}
We propose a new protocol for the generalized consensus problem 
	in asynchronous systems subject to Byzantine server failures.
The protocol solves the consensus problem in 
	a  setting in which information about conflict between transactions
	is available (such information can be in the form of transaction
	read and write sets).
Unlike most prior proposals (for generalized or classical consensus), which use a 
	leader to order transactions, this protocol is leaderless,
	and relies on non-skipping timestamps for transaction ordering.
Being leaderless, the protocol does not need to pause for leader elections.
The use of non-skipping timestamps permits servers to commit transactions
	as soon as they know that no conflicting transaction can be ordered earlier.
For $n$ servers of which $f$ may be faulty,
this protocol requires $n > 4f$.
\end{abstract}

\section{Introduction}
A \emph{distributed ledger} is a distributed data structure, replicated 
	across multiple \emph{nodes}, where \emph{transactions} from 
	clients are published in an agreed-upon total order. 
Today, Bitcoin~\cite{bitcoin} is perhaps the best-known distributed ledger 
	protocol.

There are two kinds of distributed ledgers. 
In \emph{permissionless} ledgers, such as Bitcoin, any node can 
	participate in the common protocol by proposing transactions, 
	and helping to order them. 
In \emph{permissioned} ledgers, by contrast, a node must be 
	authorized before it can participate. 
Permissionless ledgers make sense for cryptocurrencies which seek to 
	ensure that nobody can control who can participate. 
Permissioned ledgers make sense for structured marketplaces, such as 
	financial exchanges, where parties do not necessarily trust one 
	another, but where openness and anonymity are not goals. 
State machine replication~\cite{state-machine-replication} is the most 
	common way to implement permissioned ledgers.

In state machine replication, the servers agree on a total order for all transactions,
	and every server executes the transactions in the same order.
If two successive transactions commute, the two transactions can be 
	executed in different orders by different servers.
To determine if two transactions commute, we can check if 
	the state variables accessed for reading or writing (\emph{read} and \emph{write} 
	sets) by one transaction are written to by the other transactions and vice-versa.
Existing state machine replication protocols are limited in their ability to exploit 
	transaction commutativity.
Protocols that exploit general transaction commutativity solve what is called 
	the \emph{generalized consensus} problem in which a dependency 
	structure is assumed on the transactions~\cite{pedone2002handling,lamport2005generalized}.
Published work on generalized consensus,~\cite{lamport2005generalized,pedone2002handling,peluso2016making,sutra2011fast}
	with few exceptions, is limited to systems with servers subject to crash failures.
Pires et al.~\cite{10.1007/978-3-319-69084-1_14} propose a leader-based 
generalized state machine replication algorithm,
and Abd-El-Malek et al.~\cite{abd2005fault} 
	propose a client-driven quorum-based protocol called Q/U that is very efficient 
	under low contention, but that  requires $n > 5f$,
        and can suffer from livelock in the presence of 
	contention even in synchronous periods.

The contribution of this paper is a novel permissioned ledger algorithm, which we 
	call \emph{Byblos}.
Byblos has three properties of interest.

\vspace{-1ex}
\begin{itemize}
\item {\bf Generality}. 
	Byblos exploits \emph{semantic} knowledge about client requests
  		to reduce transaction latency.
	Client transactions include statically-declared read and write sets.
	The technical key to effectively exploiting semantic knowledge is a novel
		use of  \emph{non-skipping timestamp}~\cite{non-skipping} to bound the 
		set of in-flight transactions that might end up ordered before a particular 
		transaction.
	If an otherwise-complete transaction does not conflict with any of
		its potential predecessors, that transaction can be committed without
		further delay. 
	For system loads with few conflicts between transactions, solutions for generalized 
		consensus can be much more efficient than solutions for traditional 
		consensus~\cite{lamport2005generalized}.
\item {\bf Leaderless}. 
	Byblos is \emph{leaderless}.
  	With some exceptions~\cite{abd2005fault,leaderless,HoneyBadger},
		prior replicated state machine algorithms use a \emph{leader} to order 
		client requests.
	Leader-based algorithms typically have two kinds of phases:
  		a relatively simple normal phase where the leaders send and receive 
		messages to the others, and a complicated ``view change'' 
		phase~\cite{PBFT,zyzzyva,FaB,paxos-complex} used to detect and 
		replace faulty leaders.
	Leader election comes at a cost:
		client requests are typically blocked during leader election even in 
		periods of synchrony.
	Such delays are especially problematic if valuable periods of
		synchrony are spent electing leaders instead of making progress.
  		(Other leaderless protocols, such as EPaxos~\cite{epaxos}, make
  		similar observations.)
		
	In Byblos, transactions are guaranteed to terminate in periods of synchrony.
	Technically, Byblos does not need a leader because it is centered
  		around a leaderless non-skipping timestamp algorithm.
\item {\bf Simple}.  Byblos is simple to explain and understand. While simplicity is subjective,
	readers who are familiar with other protocols for Byzantine fault tolerance will note
	that the full protocol is described in this paper.
\end{itemize}

\vspace{-1ex}
In Byblos, transactions are ordered by timestamp,
with ties resolved through some canonical ordering such comparing transaction identifier hashes.
For a given timestamp value $t$, Byblos can determine an upper bound on the 
	set of in-flight \emph{pending} transactions that \emph{might} have  
	timestamp $t$ assigned to them. 
We say this ability to bound the set of potentially conflicting pending transactions makes 
	Byblos \emph{clairvoyant}.
If a transaction $T$ with timestamp $t$ is the next one to be executed by a replica 
	among those transactions with timestamp $t$, and $T$ does not conflict with 
	any of the \emph{pending} transactions, then $T$ can be executed without 
	waiting for the status of the pending transactions to be resolved.
Byblos guarantees progress by calling on an ``off-the-shelf'' asynchronous 
	Byzantine agreement algorithm, preferably early deciding~\cite{ben-or,zielinski-asynch},
	to \var{CANCEL} or \var{COMMIT} pending transactions
\footnote{Asynchronous consensus algorithms are 
	those that guarantee safety at all times, and progress under eventual synchrony}.

Byblos tolerates $f < n/4$ faulty servers, assuming the underlying consensus algorithm 
	does the same.
If there are no conflicts between pending transactions and transactions waiting for execution,
	Byblos can make progress even in periods of complete asynchrony.
(This claim does not contradict the FLP impossibility result~\cite{flp}.)
	
The rest of the paper is organized as follows. Section~\ref{sec:related} discusses related work. 
Section~\ref{sec:model} introduces the problem and the system model.
Section~\ref{sec:byblos} gives a detailed description of Byblos assuming clients fail by crashing
	and Section~\ref{sec:proof} proves its correctness and Section~\ref{sec:byzantine-client} describes 
	how Byzantine clients can be tolerated. 
Section~\ref{sec:perform} discusses the performance of Byblos under various assumptions on failures
	and synchrony and Section~\ref{sec:conclude} concludes the paper.

\section{Related Work}\label{sec:related}
Leader-based distributed ledgers such as Paxos~\cite{Paxos} and Raft~\cite{raft}
do not exploit knowledge of read-write sets to reduce latency and increase throughput.
Distributed ledgers that do exploit such information include Generalized Paxos~\cite{lamport2005generalized},
	Egalitarian Paxos~\cite{epaxos}, Hyperledger Fabric~\cite{hyperledger}, NEO~\cite{NEO},
	and Bitcoin itself~\cite{bitcoin}.

There is a large body of literature on state machine replication, most of which is 
	leader-based.
Clement \emph{et al.}~\cite{Aardvark} observe that many Byzantine
	fault-tolerant (BFT) protocols can perform poorly in the presence of Byzantine 
	failures.
They define the notion of a \emph{fragile optimization},
	where a single misbehaving party can knock the system off an 
	optimized path.
They also define \emph{gracious} (synchronous, non-faulty) and
	\emph{uncivil} (synchronous, limited Byzantine faults) executions.
They argue that while most BFT protocols are optimized only for gracious 
	executions, it is also important that protocols perform well in uncivil executions.  
They propose \emph{Aardvark}, a BFT protocol designed to perform well under 
	uncivil executions.
Aardvark uses a leader, with regularly-scheduled view changes.
The protocol includes safeguards against censorship by the leader.
Amir \emph{et al.}~\cite{Prime} introduce  \emph{bounded delay} as
	a performance goal for BFT protocols.
Bounded delay captures the inability of Byzantine servers or clients
	to impose arbitrary delays on transaction processing.
They introduced \emph{Prime},
	a BFT protocol that uses a leader that is monitored by other
	servers to provide bounded delay in the presence of limited Byzantine 
	failures.

\emph{Paxos}~\cite{Paxos} and Raft~\cite{raft} are
perhaps the best-known non-Byzantine replication protocols.
Other Paxos-related non-Byzantine protocols include
\emph{Mencius}~\cite{Mencius} and \emph{EPaxos}~\cite{epaxos}.
These protocols, with the exception of  EPaxos~\cite{epaxos}, use some form 
	of leader (or leaders) and view changes.
Milosevic \emph{et al.}~\cite{BoundedByz}  proposed a \emph{BFT-Mencius}
	which also uses performance monitoring and view changes 
	to limit the effects of slow servers.
Byblos does not use view changes or performance monitoring and
	hence allows unbounded variance below the timeout threshold.

Prior protocols that perform relatively well under uncivil executions
	perform less well in civil executions,
	compared to protocols optimized only for civil executions.
Byblos is different.
Its latency, measured in the number of message round trips,
	is comparable to protocols optimized for civil execution.
In the absence of slow or faulty clients, 
	its latency in uncivil executions is also comparable to that of
	protocols optimized for civil executions.
On the downside,
	Byblos uses signatures, whereas some protocols use faster message 
	authentication codes.

One BFT protocol that does not use leaders or view changes is	
	\emph{HoneyBadgerBFT}~\cite{HoneyBadger}.
Unlike most BFT protocols,
HoneyBadgerBFT does not assume eventual (or partial)
	synchrony, but relies on
	a randomized atomic broadcast protocol with a cryptographic shared coin.
HoneyBadgerBFT ensures censorship resistance through a cryptographic subprotocol.
Unlike Byblos, HoneyBadgerBFT does not exploit transaction semantics.
The \emph{RBFT} BFT protocol~\cite{rBFT} uses multiple
	leaders, who track one another, and provide censorship resistance.
It is designed for 
systems in which clients 
	can have multiple parallel pending requests.
Aublin \emph{et. al}~\cite{AublinLGKQV2015} describe a family of protocols,
	some of which have low (2-message) latency in synchronous executions.

As noted, the protocols discussed, with the exception of  EPaxos~\cite{epaxos}
	which only tolerates crash failures, do not solve the generalized
	consensus problem~\cite{lamport2005generalized,pedone2002handling}.
Abd-El-Malek et al.~\cite{abd2005fault} 
	propose a client-driven quorum-based protocol called Q/U that is very
	efficient under low contention, but that  requires $n > 5f$ and 
	can suffer from livelock due to contention even in synchronous periods.
The algorithm is leaderless and uses exponential backoff in the presence of contention.
Cowling \emph{et al.}~\cite{cowling2006hq} aims at improving Q/U by reverting to using a leader.
Recently Pires et al.~\cite{10.1007/978-3-319-69084-1_14} proposed a leader-based Byzantine version 
	of generalized Paxos. 


In general, faulty clients in Byblos can force servers to revert to
an ``off-the-shelf'' binary Byzantine consensus protocol to resolve
the outcome of ``stuck'' transactions.
Triggering the agreement protocol might incur a timeout which can be
significantly larger than typical communication delay even for fast protocols
	(for example, Ben-Or \emph{et al.}~\cite{Ben-OrKR1994} or Mostefaoui
	\emph{et al.}~\cite{MostefaouiMR2014}).
It might seem that Byblos replaces one source of delay
	(a faulty leader) with another (faulty clients),
but this replacement allows us to exploit transaction semantics which 
	can be a significant improvement in some settings.
In systems in which faulty servers can delay the processing of
	transactions (which is almost all systems), everyone is delayed.
	(These issues are discussed in Section~\ref{sec:perform}.)

\section{Problem and System Model}\label{sec:model}
A \emph{ledger} (Figure~\ref{fig:ledger}) can be thought of as an
	automaton consisting of a set of \emph{states}
	(for example, clients' account balances),
	a set of deterministic state transitions called \emph{transactions}
	( for example, deposits, withdrawals, and transfers),
	and a \emph{log} recording the sequence of transactions.
The state is needed to efficiently compute transactions' return values
	(for example, overdrawn account).
The log provides an audit trail:
	one can reconstruct any prior state of the ledger,
	and trace who was responsible for each transaction.

Our solution encompasses the following components.
There are $n$ \emph{servers} that maintain the ledger's long-lived state via a
	set of replicated  logs. 
Up to $f$ of $n = 4 f + 1$ servers may be Byzantine
	(not behaving according to the protocol).
The rest of the servers are \emph{correct}.
The logs of correct servers are only modified by appending new transactions.
The servers satisfy the following safety property: for any pair of correct servers,
	one server's log is a prefix of the other's.
It follows that correct servers execute all transactions 
	in the same order.
	
There is a potentially unbounded number of \emph{clients}
	who originate transactions.
It is the servers' job to accept transactions from clients,
	order them, and publish this order.
We assume that the clients are not Byzantine;
	in Section~\ref{sec:byzantine-client} we explain how to handle Byzantine clients.

Communication is handled by an underlying message-diffusion system.
Clients broadcast messages,
	which are eventually delivered to all correct servers.
Servers communicate with one another through the same diffusion infrastructure.
We assume that messages are delivered in the order they are sent which can be easily 
	implemented in practice using message counters.
We assume that the identity of the sender is included in the message and that messages cannot be forged;
a client or server cannot send a message to another server or client that has a different sender than the actual sender.
In practice,  this is enforced using signatures.

The ledger state is a key-value store.
Each client transaction declares a \emph{read set},
	the set of keys it might possibly read, and a \emph{write set},
	the set of keys it might possibly write.
Any transaction that violates its declaration is rejected.
	(Systems such as Generalized Paxos~\cite{lamport2005generalized},
Egalitarian Paxos~\cite{epaxos}, and NEO~\cite{NEO} all make use of similar 
	conflict declarations.)

\begin{figure}[!t]
\begin{lstlisting}
  state = initialState
  log = []
  while true:
    on receive T from c:
      log.append(T)
      state, result = apply(T, state)
      send result to c
\end{lstlisting}
\caption{Ledger Abstraction}
\label{fig:ledger}
\end{figure}

\section{Byblos Description}\label{sec:byblos}

In Byblos, transactions are assigned integer \emph{timestamps},
which partially determine the order in which transactions are applied.
If two transactions do not overlap in time,
the later one will be assigned a larger timestamp, but overlapping
transactions may be assigned the same timestamp.
One can think of assigning a timestamp as assigning a slot to a transaction.
Non-overlapping transactions will be in different slots, but overlapping transactions
	might be in the same slot.
The protocol we use for assigning timestamps guarantees that a new slot 
	will not be assigned to a transaction before  every previous slot has at least
	one transaction in it. 
This is guaranteed by {\em non-skipping} property of the assigned timestamps~\cite{non-skipping}: 
	if a timestamp $t$ is assigned to a transaction, then every timestamp whose
value is less than $t$ must have been previously assigned to some other transaction.

The non-skipping timestamp protocol~\cite{non-skipping} at
the heart of the algorithm is simple.
A client broadcasts a timestamp request to the servers,
and collects at least $n-f$ timestamps in response.
The client selects the $(f+1)^\textrm{st}$ latest timestamp,
which is guaranteed to be less than or equal to the latest timestamp
assigned to any transaction.
The client increments that timestamp by one,
and later broadcasts it to the servers.
It has been shown~\cite{non-skipping} that this way of choosing
timestamps ensures that no timestamp value is skipped.

The properties of non-skipping timestamps suggest a simple way for servers to 
execute transactions in a deterministic order in the absence of client failures.
Starting with timestamp 1, execute all transactions with timestamp 1 in some 
deterministic order, then execute all transaction with timestamp 2 in some 
deterministic order, and so on. In general,  once all transactions with timestamp 
$t$ are executed, transactions with timestamp $t+1$ can be executed. 
This technique is, of course, too simple, even in the absence of client failures.
The catch is that servers cannot determine when all
transactions with a given timestamp value have been received.

To determine when all transactions with timestamp $t$ have been received,
Byblos calculates for each timestamp $t$, a set of \var{pending}[$t$] of 
transactions that were detected to be concurrent with, or 
occurring before transactions with timestamp $t$. 
The set of pending transactions contains transactions, but not their assigned 
timestamps, because those timestamps might not be known at the time a 
transaction is added to a pending set. The crucial property is the following:  
\var{pending}[$t$] is guaranteed to contain all transactions whose assigned 
timestamp will be $t$ or less. However, it may also include some transactions 
that will be assigned timestamps larger than $t$. For a given transaction $T$ 
with timestamp $t$, if the timestamps of the transactions in \var{pending}[$t$] 
are known, then, by the property of the pending set, the set of all transaction 
that could be assigned timestamp $t$ is also known. In the absence of client 
failures, servers can execute transactions  as follows: transactions with 
timestamp $t$ can be executed when the timestamps of all transactions  
in the pending set become known and all transactions ordered earlier
have been executed. 

Servers can do better by taking into consideration potential conflicts between 
transactions. Given a deterministic ordering on transactions, with transactions with 
smaller timestamps appearing before transactions with larger timestamps in the 
order, a transaction $T$ with timestamp $t$ can be executed when all conflicting 
transactions that appear before $T$ in the order are executed and there are no 
conflicting transactions in the pending set that could be ordered before $T$ (this 
requires that the timestamps of those conflicting transactions be known).

The description so far assumes no client failures. If clients can fail, some transactions 
in the pending set might never be assigned timestamps, and the servers will be stuck, unable to 
determine when all potentially conflicting transactions with timestamp $t$ have been 
received. We resolve this situation by falling back to a binary consensus algorithm, over 
the values \var{COMMIT} and \var{CANCEL}, to resolve the fates of orphaned 
transactions. Each client tries to commit its own transaction using the consensus 
algorithm, and servers try to cancel pending transactions that are slow to receive their timestamps.
We can use any ''off-the-shelf'' consensus algorithm that is 
guaranteed to terminate if the system is eventually synchronous, including known 
algorithms that terminate in one round if the system is 
well-behaved~\cite{zielinski-asynch}. Transaction execution proceeds as follows. 
Once all conflicting transactions in the pending set of some transaction $T$ with 
timestamp $t$ are either cancelled or their timestamp is known, the position of $T$ 
in the order will be known and the execution can then proceed as outlined above 
for all transactions that have not been cancelled.

The protocol guarantees safety at all times and liveness under eventual 
synchrony~\cite{PBFT}. 
The rest of this section describes the client and server code in details. In this 
description we assume that clients can fail by crashing but are not Byzantine. 
We address Byzantine clients in Section~\ref{sec:byzantine-client}.

\begin{figure}[!t]
\begin{lstlisting}
 received = `$\emptyset$` // messages from servers
  `$\forall$` s in Servers timestamp[s] = 0

  // get valid timestamp
  broadcast Propose(T) to Servers`\label{line:propose}`
  repeat
     on receive m: received = received `$\cup$` {m}
  until ProposeAck(T, t) received from `$\geq n-f$` servers `\label{line:propose-ack}`
   `$\forall$` s in Servers: if `$\exists$` t : ProposeAck(T, t) received from s
       timestamp[s] = t
   `$\hat{t} = ((f+1)^{\textrm{st}}$` largest value in timestamp[*]) + 1`\label{line:ts}` `\label{line:propose:end}`

  // broadcast timestamp, and wait for transaction result
  broadcast Confirm(T,`$\hat{t}$`) to Servers`\label{line:confirm}`
  repeat
     on receive m: received = received `$\cup$` {m}
  until `$\geq f+1$` identical Resolve(T, code, result) messages received 
  if code = COMMIT (in `$\geq f+1$` messages) then
     return result
  else:
     return `$\bot$``\label{line:commit:end}`
\end{lstlisting}


\caption{Client Code}
\label{fig:client}
\end{figure}

\subsection{Client Code}

The client code (Figure~\ref{fig:client}) proceeds in two stages.
In the first stage (Lines~\ref{line:propose}--\ref{line:propose:end}), the client sends 
a \kPropose{}  message to all servers, and collects at least $n-f$ \kProposeAck{} 
responses. Each response contains an integer timestamp that is the local 
\var{clock} value at the server when the server received the client request. 
The client calculates $\hat{t}$ which is equal to 1 plus the $(f+1)^\textrm{st}$ 
largest among the timestamps it received and assigns it to its transaction. 
It is important to note that this particular way of choosing timestamps is what 
guarantees timestamps to be non-skipping.
In the second stage (Line~\ref{line:confirm}--\ref{line:commit:end}), the client 
broadcasts a \kConfirm{} message with $\hat{t}$,and waits to receive $f+1$ 
identical \kResolveAck{} responses to determine the transaction's outcome. 
A \kResolveAck{} message has three fields:
(1) the transaction, (2) a code, either \var{COMMIT} or \var{CANCEL}, and 
(3) a result. If the return code is \var{COMMIT},
the call was successful, and the result is returned, otherwise a failure indication 
is returned. If the return code is \var{CANCEL}, the call was unsuccessful.
perhaps because the client was delayed in the middle of the protocol.
If a  failure indication $\bot$ is returned, the client is free to attempt the call 
again.

\subsection{Server Code}
\begin{figure}[!t]
\begin{center}
\begin{tabbing}
xxxxxxxxxxxx\=xxxxxxxxxxxxxxxxx\=\kill
\> \var{state} = initialState		\\
\> \var{clock} = 0	\\
\> \var{proposed}[*][*] = $\emptyset$     \>   // servers $\mapsto$ timestamp $\mapsto$ Set transaction                                      \\
\> \var{pending}[*] = $\emptyset$    \>  // timestamp $\mapsto$ Set transaction		\\
\> \var{confirmed}[*] = $\emptyset$ \>   // timestamp $\mapsto$ Set transaction		\\
\> \var{committed} = $\emptyset$    \>   // timestamp $\mapsto$ Set transaction		\\
\> \var{cancelled} = $\emptyset$     \>  // Set transaction					\\
\> \var{resolving}  = $\emptyset$     \> // Set transaction					\\
\> \var{log} = []          \>  // Sequence transaction\\
\> \var{timer[*]} = $\infty$          \>  // timestamp $\mapsto$ real time timer\\
\> \var{time} : real time clock value at a server	
\end{tabbing}
\end{center}
\caption{Server State with Initializations}\label{fig:server-data}
\end{figure}

\subsubsection{Server State}
The server state (Figure~\ref{fig:server-data} ) is composed of the following fields.
\begin{itemize}
\item
\var{state} is the ledger state. A transaction is applied to \var{state} when it 
commits.

\item
\var{clock} is an integer counter that tracks the latest timestamp assigned to
a transaction. We assume this counter does not overflow. Since timestamps
are non-skipping, a 128-bit counter should be more than sufficient in practice.

\item
\var{proposed} keeps track of transactions that have been proposed at various 
servers and the local clocks at those servers when the proposed message 
was received. When a transaction is added to \var{proposed}[$s$][$\mit{k}$], 
where $s$ is a server and $\mit{k}$ is a clock value, the timestamp of the 
transaction might not be known.

\item
\var{pending} is a map from timestamps to sets of transactions. For timestamp $t$,
\mbox{\var{pending}[t]} is the set of transactions that might be assigned timestamp $t$ or earlier. 

\item
\var{confirmed} is a map from timestamps to sets of transactions. For timestamp 
$t$,  \mbox{\var{confirmed}[t]} is the set of known transactions that will either commit 
with timestamp $t$ or will be cancelled. 

\item
\var{committed} is set of transactions known to have committed.
  
\item
\var{cancelled} is the set of transactions known to be cancelled.

\item
\var{log} is the sequence of committed transactions as they are applied 
to the ledger.

\item 
\var{timer} is an array of timers used to 
	timeout pending transactions. 

\item 
\var{time} is a local time at the server to measure
real time for timeouts. The local time of servers are independent and 
need not be synchronized.
\end{itemize}

\subsubsection{Server Actions}
The server continually receives messages (Figure~\ref{fig:server}). In the pseudocode 
we denote the sender of a message received by a server with $c$ if the sender is a 
client and with $s$ if the sender is a server.

When the server receives a \kPropose(\var{T}) (Lines~\ref{line:proposeAck}--\ref{line:proposeAck:end}) 
message from client \var{c}, it adds (\var{T},\var{clock}) to \var{proposed}[\var{self}], 
where \var{self} denotes the server executing the code and \var{clock} is the local clock 
value when the proposed message is received. The server returns the current \var{clock} 
value to the client and forwards the proposed message together with the clock value at which 
it was received to all other servers. When a server receives a forwarded \kPropose{} message 
from server \var{s}, it adds the message to the set \var{proposed}[\var{s}] of proposed messages 
that server \var{s} is aware of. These sets of messages will be used in calculating the sets of 
pending transactions for any timestamp value.

When a server receives a  \kConfirm(\var{T},$\hat{t}$) (Line~\ref{line:confirmAck}) message 
from a client or a server, the server advances \var{clock} to the maximum of $\hat{t}$ and 
its current value, and adds the transaction to the set of confirmed transactions for timestamp 
$\hat{t}$ and forwards the \kConfirm{} message to the other servers. If the received 
\kConfirm{} message is from a server \var{s}, \var{s} is added to a set 
\var{ConfirmWitness}[\var{T}], which is the set of servers that has received the 
\kConfirm{} message for transaction \var{T}.
If the set of witnesses is large enough, and the set \var{pending}[$\hat{t}$] 
is not yet calculated, the server calculates the set \var{pending}[$\hat{t}$]
by taking the union of all propose messages that servers in the set of 
witnesses were aware of before receiving the \kConfirm{} message for \var{T} . 
We will show that this pending set is guaranteed to contain all transactions that 
can be confirmed with timestamp $\hat{t}$ or less. 
Finally, if the transaction \var{T} for which the \kConfirm{} message is received is not 
in the process of being resolved (not resolving), the server also launches a consensus 
protocol with the other servers to try to to COMMIT \var{T}. The consensus protocol is run in
a \var{ResolveThread} forked by the server (Figure~\ref{fig:resolve}).

\begin{figure}[!t]
\begin{lstlisting}
repeat
   if Propose(T) received from c:						`\label{line:proposeAck}`
      proposed[self] = proposed[self] `$\cup$` {(T,clock)}
      send Proposed(T, clock) to all servers
      send ProposeAck(T, clock) to c 					`\label{line:proposeAck:end}`

   if Proposed(T,clock) received from s:
      proposed[s] = proposed[s] `$\cup$` {(T,clock)}

   if Confirm(T, `$\hat{t}$`) received from c or s:				`\label{line:confirmAck}`
      clock = max(clock, `$\hat{t}$`)				`\label{line:clock}`    
      confirmed[`$\hat{t}$`] = confirmed[`$\hat{t}$`] `$\cup$` {T}
      send Confirm(T, `$\hat{t}$`) to all servers if not previously sent
      ConfirmWitness[T] =  [T] `$\cup$` {s}
      if pending[`$\hat{t}$`] `$= \emptyset$` `$\wedge$` |ConfirmWitness[T]| `$= n-f$`:   
           pending[`$\hat{t}$`] =  `$\bigcup_{s \in \mbox{ConfirmWitness[T]}} \{$`T': `$\exists\,\mit{k}\leq \hat{t},$` (T,k) `$\in$` proposed[s]`$\}$``\label{line:pending-update}`
           timer[`$\hat{t}$`] `$=$` time`$+ \delta$` `\label{line:delay}`
      if T `$\not\in$` resolving:
          resolving = resolving `$\cup$` {T}
          fork ResolveThread(T, `$\hat{t}$`, COMMIT)         // try to commit this transaction `\label{line:COMMIT}`


   if StartResolution(T, `$\hat{t}$`, code) received from s:  // join another consensus`\label{line:startResolution}`
      if T `$\not\in$` resolving:                         // code is either COMMIT or CANCEL
          if code = COMMIT:
              resolving = resolving `$\cup$` {T}
              fork ResolveThread(T, `$\hat{t}$`, COMMIT);            
          else if code = CANCEL:                        
              CancelWitness[T] = CancelWitness[T] `$\cup$`  {s}  
              if |CancelWitness[T]| `$\geq f+1$`:     // correct server wants to cancel
                  resolving = resolving `$\cup$` {T}
                  fork ResolveThread(T, `$\bot$`, CANCEL);   	`\label{line:startResolution:end}`      	
     
                  
   if timer[`$t$`] expired for some `$t$`:                // if timer expired, try to cancel 
      for each txn in pending[t]:                //  slow transactions
          if txn `$\not\in$` resolving:                  
            resolving = resolving `$\cup$` {txn}
            fork ResolveThread(txn, `$\bot$`, CANCEL)				`\label{line:resolve:end}` 

   ApplyResolvedTransaction()                    // attempt to apply transactions
until false
\end{lstlisting}

\caption{Server Code}\label{fig:server}  
\end{figure}

\begin{figure}[!ht]
\begin{lstlisting}
ResolveThread(T, `$t$`, code)
   send StartResolution(T, `$t$`, code)  to all servers   `\label{line:start-resolution}` 
   if code = COMMIT then
      initialValue = 1
   else 
      initialValue = 0
   
   // run consensus protocol for transaction T 
   // with initial value = initialValue
   
   if decision = 1 then
      add T to committed
   else
      add T to cancelled
\end{lstlisting}
\caption{Resolve Thread}\label{fig:resolve}  
\end{figure}

\begin{figure}[!ht]
\begin{lstlisting}
  OrderBefore(T,T') = 
        T `$\in$` confirmed[`$t$`] `$\wedge$`
        ((T' `$\in$` confirmed[`$t'$`] `$\wedge$` `$t < t'$`) `$\vee$` (T' `$\in$` confirmed[`$t$`] `$\wedge$` `$T < T'$`) `\label{order-before}`
    
  ApplyResolvedTransaction() 
     if `$\exists$` T `$\in$` committed `$\cap$` confirmed[`$t$`] `$\wedge$` pending[`$t$`] `$\neq \emptyset$` : `$\forall$`T' `$\in$` pending[`$t$`]:
              `$\neg$`conflict(T,T') `$\vee$` T' `$\in$` cancelled `$\vee$` T' `$\in$` log `$\vee$` OrderBefore(T,T') `\label{line:termination}`
          log.append(T)
          state, result = apply(T, state))   `\label{line:apply}`
          send ResolveAck(T,COMMIT,result) to T sender
     if `$\exists$` T: T `$\in$` cancelled `$\wedge$` T `$\not\in$` log  
          log.append(T)
          send ResolveAck(T,CANCEL,`$\bot$`) to T sender 
\end{lstlisting}
\caption{Applying Resolved Transactions}\label{fig:apply}  
\end{figure}

The only remaining point that can obstruct \var{T}'s execution are pending transactions. 
These are transactions that might be confirmed with timestamps less than
or equal to the timestamp of \var{T} and,
if they conflict with \var{T},
will be applied to the log before \var{T}.
The server sets a timer to give each pending transaction time to be confirmed.
If the timer expires and a transaction in the pending set is not confirmed, the transaction
is considered obstructing,
and a starting a consensus protocol is launched to try to \var{CANCEL} the the transaction. 
In practice, timer delay can be adjusted dynamically to match the current level of synchrony~\cite{FaB}.
We discuss the timer requirements later in the correctness proofs.

The \var{ResolveThread} (Figure~\ref{fig:resolve}) encapsulates the consensus protocol executed by the servers. 
The first message sent by ResolveThread is a \var{StartResolution(T,code)} message (Line \ref{line:start-resolution} of \var{ResolveThread})
	which lets a server that has not heard directly from a client join the consensus protocol for a given transaction 
	(Lines~\ref{line:startResolution}--\ref{line:startResolution:end} of the server code). 
After sending the \var{StartResolution} message, the \var{ResolveThread} starts a binary consensus 
	protocol for the transaction with initial value 1 for COMMIT and initial value 0 for CANCEL.
The ResolveThread adds $T$ to the set \var{committed} if the decision of the consensus protocol is 1 and adds 
	$T$ to the set \var{cancelled} if the decision of the consensus protocol is 0.

        When a server receives a \var{StartResolution(T,code)},
        its behavior depends on whether \var{code} is \var{COMMIT} or \var{CANCEL}.
A server immediately joins an attempt to commit a transaction,
but joins an attempt to cancel a transaction if it receives $f+1$ distinct
\var{StartResolution(T,CANCEL)} messages for transaction $T$,
ensuring that at least one correct server timed out that transaction.

Finally, a server attempts to apply transactions (Figure~\ref{fig:apply}) to the replica state.
It is essential that conflicting transactions are executed in the same order by all servers.
We use the \var{OrderBefore}() predicate which to determine the order of confirmed transactions 
(whose timestamps are known). 
Servers order confirmed transactions first by timestamp,
and then by hashing their transaction identifiers,
yielding a deterministic order on transactions,
called the \emph{canonical} order
(Figure~\ref{fig:apply}, Line~\ref{order-before}).
If any transaction is not confirmed,
the \var{OrderBefore}() predicate evaluates to \emph{false}.

A transaction can be applied once is committed and confirmed for timestamp $t$,
the set of pending transactions for $t$ is known,
and if every other transaction \var{T'} in the pending set is either non-conflicting,
cancelled, or already applied to the log.
	
Note that it is possible that pending transactions
might appear in groups with different timestamps at different correct servers, but if they become committed, they 
will have the same timestamp at all correct servers, and if they are cancelled, they will be cancelled by all correct servers. 

\section{Proofs of Correctness}\label{sec:proof}
In this section we prove that the solution correctly implements a
distributed ledger by showing that client transactions receive the
same responses from all correct replicas,
and that in periods of synchrony the non-faulty clients' transactions
are eventually applied. 

We start by restating the model assumptions.
Messages are delivered in the order sent:
if a correct server $s_1$ receives a
forwarded \kConfirm{} message from another correct server $s_2$, then
$s_1$ must have received all previously forwarded \kPropose{} messages
from $s_2$.
In periods of synchrony,
the binary consensus protocol is guaranteed to terminate.  
All correct servers invoke \var{ResolveThread}{}),
and either they all add the transaction to the \var{committed} set,
or they all add the transaction to the \var{cancelled} set.
Finally, client messages forwarded by servers are signed by the
clients and cannot be forged.
(We omit message validation from our code for simplicity.)
	
\subsection{Supporting Lemmas}
\begin{lemma}[Same confirmed timestamp]
If a transaction is confirmed with timestamps $t$ and $t'$ at two correct servers, then $t = t'$.
\end{lemma}
\begin{proof}
A transaction is added to a the confirmed set if a \var{Confirm}(\var{T},$\hat{t}$) is received from a client directly or indirectly forwarded by a server (Figure~\ref{fig:server}  Line~\ref{line:confirmAck}). Since messages by the client are assumed to be signed by the client and clients are assumed to be non-Byzantine, a server cannot forge a confirm message and the confirm message sent by the client to different servers will have the same timestamp.
\end{proof}

The previous lemma allows us to define what it means for a transaction to be assigned a timestamp value.
\begin{definition}[Assigned timestamp]
We say that a transaction $T$ is \emph{assigned timestamp} $t$ if some correct server received 
a \var{Confirm}(T, $\hat{t}$) message from a client.
\end{definition}

\begin{lemma}[Non-Skipping timestamps]
If a transaction is assigned timestamp $t$, then for every $t'$, $0 < t' < t$, some transaction 
	is assigned timestamp $t'$.
\end{lemma}
\begin{proof}
We show that for any $t_0$, the first timestamp larger than $t_0$ assigned to a transaction is $t_0+1$.
Consider the first transaction $\mbox{\var{T}}'$  for which 
a timestamp $t' > t_0$ is assigned. The client must have received replies greater than or equal $t'-1$ from 
$f+1$ servers one of which must be correct. 

Consider the correct server that returned a timestamp $t'' \geq t'-1 \geq t_0$.
The timestamp $t''$ must have been assigned to an earlier transaction because the clock 
	can only advance at a correct server after a timestamp is assigned (Line~\ref{line:clock}).
This means that $t''$ cannot be greater than $t'-1$ for otherwise an earlier transaction is assigned timestamp larger than $t_0$. But $t'-1 \geq t_0$, which gives us $t'' \leq t_0$.
So, we have established that  $ t_0 \leq t'-1 \leq t''$  and $t'' \leq t_0$, which means that $t'' = t_0$ and $t' = t_0+1$.
 \end{proof}

The next two lemmas establish that if $\mbox{\var{T}}_1$  does not \emph{see} $\mbox{\var{T}}_2$, then $\mbox{\var{T}}_2$ must have a larger timestamp and
executing $\mbox{\var{T}}_1$  can be done without establishing the resolution status of $\mbox{\var{T}}_2$. We start with a definition of what it means for a transaction to see another.

\begin{definition}[Transaction seeing another]
We say a transaction $T_1$ \emph{sees} that is assigned timestamp $t_1$ sees transaction $T_2$ if $T_2$ is
added to the \var{pending}[$t_1$] set of some correct server. 
\end{definition}

\begin{lemma}[Non-concurrent transactions get different timestamps]\label{lemma:see}
If $\mbox{\var{T}}_1$ does not see $\mbox{\var{T}}_2$ and $t_1$ is assigned a timestamp,
then $\mbox{\var{T}}_2$ is assigned a later timestamp.
\end{lemma}
\begin{proof}
Suppose $\mbox{\var{T}}_1$  and $\mbox{\var{T}}_2$  are assigned timestamps $\hat{t}_1$ and  $\hat{t}_2$ respectively.
When a correct server receives $\mbox{\var{T}}_1$'s \kConfirm{} message, its clock value will be at least $\hat{t}_1$ (Line~\ref{line:clock}). 
Since $\mbox{\var{T}}_1$ is assigned a timestamp, its \kConfirm{}
message is received by some correct server,
and eventually all correct server will receive $\mbox{\var{T}}_1$'s
\kConfirm{} message.
It follows that that there will be enough witnesses to
$\mbox{\var{T}}_1$'s \kConfirm{} message and \var{pending}[$t_1$] will
be calculated by every correct server (Line~\ref{line:pending-update}), 
When \var{pending}[$t_1$] is calculated, a set of $n-f$ servers, which we will call $S_1$, has confirmed the  
receipt of $\mbox{\var{T}}_1$'s \kConfirm{} message with timestamp $\hat{t_1}$.  
Because $\mbox{\var{T}}_1$  did not see $\mbox{\var{T}}_2$ ,
$\mbox{\var{T}}_2$'s \kPropose{} message could not have preceded $\mbox{\var{T}}_1$'s
\kConfirm{} message at any correct sever in $S_1$.

Let $S_2$ be the set of servers that acknowledge $\mbox{\var{T}}_2$'s \kPropose{} messages. 
Every correct server in $S_2$,  added $\mbox{\var{T}}_2$  to the \var{proposed} set and returned its
	local clock value to the client (Line~\ref{line:propose-ack}).
Because $\mbox{\var{S}}_1$  and $\mbox{\var{S}}_2$  both have size at least $n-f$,
their intersection includes at least $f + 1$ correct servers.
At each of these servers, $\mbox{\var{T}}_2$'s \kPropose{} message
arrived after $\mbox{\var{T}}_1$'s \kConfirm{} message,
so $\mbox{\var{T}}_2$  received at least $f+1$ \kProposeAck{} messages with timestamp
values greater than or equal to $\hat{t}_1$,
so $\mbox{\var{T}}_2$  will choose $\hat{t}_2$ greater than $\hat{t}_1$.
\end{proof}

It follows that, for a given transaction $\mbox{\var{T}}$  with timestamp $\hat{t}$, the set $\mbox{\var{pending}}[t]$ at a correct server contains only transactions that can be assigned timestamp $\hat{t}$ or less because any transaction not in the set will be assigned a timestamp larger than $t$.  But we need to be sure that if $\mbox{\var{pending}}[t]$ at a given correct server is not empty, then it will contain all transaction that could ever be in $\mbox{\var{confirmed}}[t']$, $t' \leq t$. This would allow the server to detect when all conflicting transactions with timestamp $t$ or less have been resolved (\var{COMMIT} or \var{CANCEL}), at which time it can execute those among them that have not been cancelled. The following lemma establishes this fact.

\begin{lemma}[\var{pending} has all the information needed to resolve transactions]
At all times, $\bigcup_{t' < t} \mbox{\var{confirmed}}[t']  \subseteq P$, where $P = \mbox{\var{pending}}[t]$ is a non-empty $\mbox{\var{pending}}[t]$ set calculated by a correct server.
\end{lemma}
\begin{proof}
The lemma states that any transaction that will ever be confirmed with timestamp less than or equal to $t$ must be in every $\mbox{\var{pending}}[t]$ set calculated by a correct server. The proof is similar to the proof of~\ref{lemma:see}.
Let \var{T} be a transaction that is assigned timestamp $t$ and $\mbox{\var{pending}}[t]$ by the pending set calculated by a correct server on the receipt of \var{T} \kConfirm{} message when there are enough witnesses to
the \kConfirm{} message. If a transaction \var{T'} will be assigned timestamp $t' \leq t$, then at most $f$ servers reply to the client's \kPropose{} message with timestamp value greater than $t'-1$. Of the remaining $n-2f$ servers, at most $f$ are faulty, which leaves $f+1$ correct servers that reply to the propose message with a timestamp less than $t'$. 
One of these $f+1$ servers must forward \var{T}'s \kConfirm{} message and must have forwarded the \kPropose{} message before that. It follows that  $\mbox{\var{pending}}[t]$ must contain transaction \var{T}.
\end{proof}



\subsection{Safety}

We start by showing that all correct servers execute conflicting
transactions in the same order,
and that all correct servers execute the same transactions.
Then we show that the execution at each correct
server is equivalent to a ``canonical'' execution in which transactions
(even non-conflicting ones) are executed according to timestamp order
(with ties broken canonically).
This implies that correct servers return the same result for each transaction.
Then we show that the execution is linearizable by showing that it is
equivalent to an execution in which each transaction is executed
atomically at some point between its invocation and response.

\begin{lemma}[Same order for applied conflicting transactions]
If two correct servers apply two conflicting non-cancelled  transactions $\mbox{\var{T}}_1$ and $\mbox{\var{T}}_2$  to the log, they apply them in the same order.
\end{lemma}
\begin{proof}
If $\mbox{\var{T}}_1$ and $\mbox{\var{T}}_2$ are two conflicting
non-cancelled transactions and $\mbox{\var{T}}_1$ is applied before
$\mbox{\var{T}}_2$, then
\var{OrderBefore}($\mbox{\var{T}}_1$,$\mbox{\var{T}}_2$) must be \emph{true},
implying that that $\mbox{\var{T}}_1$'s assigned  timestamp  is
smaller than that of $\mbox{\var{T}}_2$,
or that both are assigned the same timestamp and $\mbox{\var{T}}_2$
appears before $\mbox{\var{T}}_2$ in the canonical order.
This property holds at every correct server,
so all correct servers that apply both transactions will apply them in the same order.
\end{proof}

\begin{lemma}[Agreement on  transaction resolution]
If a correct server decides to commit or cancel a transaction,
then every correct server eventually makes the same decision.
\end{lemma}
\begin{proof}
The decision to commit or cancel a transaction depends on the decision
of the consensus protocol initiated in the \var{ResolveThread} (Figure!\ref{fig:resolve}).
By the agreement property of consensus, any two correct servers will have the same decision value.
\end{proof}

\begin{definition}[Canonical execution]
A \emph{canonical execution} is one that contains all transactions
applied to the log of some correct server ordered by first by
timestamp, then by transaction hash code.
Cancelled transactions with unknown timestamp are assumed to be
assigned the smallest timestamp for which they appear in a pending set of a correct server. 
\end{definition}
Since all correct servers execute the same sequence of transactions,
this definition is independent of the particular correct server chosen.

\begin{lemma}[Every execution is equivalent to the canonical execution]
The responses of transactions as applied to the log of a correct server is identical to the responses of the same transactions as they appear in the canonical execution.
\end{lemma}
\begin{proof}
We consider an execution $E$ of a correct server in which transactions are applied as allowed by the protocol.
We say that two transactions $\mbox{\var{T}}_1$ and $\mbox{\var{T}}_2$
are \emph{out-of-order} if $\mbox{\var{T}}_1$ appears before
$\mbox{\var{T}}_2$ in $E$ but $\mbox{\var{T}}_2$ appears before
$\mbox{\var{T}}_1$ in the canonical execution. We prove by induction
on the number of out-of-order pairs of transactions that $E$ is
equivalent to the canonical execution. For the base case, there are no
out-of-order transactions and $E$ is the canonical execution. For
the induction step, since $E$ has at least one out-of-order pair
of transactions, it must have to successive transactions, say
$\mbox{\var{T}}_1$ and $\mbox{\var{T}}_2$, that are out of
order. Consider the execution \var{E'} which is identical to var{E'}
except that $\mbox{\var{T}}_2$ appears before $\mbox{\var{T}}_1$ in
\var{E'}. We show that  
\var{E'} and $E$ are equivalent. Also, compared to $E$, \var{E'} has exactly one less pair of transaction that are out-of-order.  It follows by the induction hypothesis \var{E'} is equivalent to the canonical execution, so $E$ is also equivalent to the canonical execution.

Since  $\mbox{\var{T}}_1$ and $\mbox{\var{T}}_2$ are out of order in $E$, it follows that either both are committed and do not conflict or at least one of them is cancelled. If one of them is cancelled, it is immediate that \var{E'} and $E$ are equivalent. If $\mbox{\var{T}}_1$ and $\mbox{\var{T}}_2$ do not conflict, we divide \var{E'} into three parts: 
(1) $E_b$ of transactions that appear before
$\mbox{\var{T}}_1$ and $\mbox{\var{T}}_2$, (2) $\mbox{\var{T}}_2$
followed by $\mbox{\var{T}}_1$, and (3) $\mbox{$E$}_a$ of transactions
that appear after $\mbox{\var{T}}_1$ and $\mbox{\var{T}}_2$. The
execution of $\mbox{$E$}_b$  is the same in $E$ and \var{E'}, so the
state of the ledger is identical in $E$ and \var{E'} after executing
$\mbox{$E$}_b$. Since $\mbox{\var{T}}_1$ and $\mbox{\var{T}}_2$ do not
conflict, they access and modify different parts of the state (read
and write sets are disjoint), so executing $\mbox{\var{T}}_2$ followed
by $\mbox{\var{T}}_1$ after $\mbox{$E$}_b$ will result in the same
ledger state and same responses as executing $\mbox{\var{T}}_1$
followed by $\mbox{\var{T}}_2$ $\mbox{$E$}_b$. Finally $\mbox{$E$}_a$
is identical in $E$ and \var{E'} and executed starting from the same
ledger state, so all transactions in $\mbox{$E$}_a$ will return the
same response. 
\end{proof}

\begin{theorem}[Linearizability]
The implementation is linearizable.
\end{theorem}
\begin{proof}
Since all applied transactions are committed and correct servers commit the same set of transactions, 
	it follows that the sets of transactions applied by correct servers are the same. 
The previous lemma shows that the responses of transactions applied by the correct servers is the same.
It remains to define an execution that is consistent with real time order and consistent with the executions of 
	correct servers and such that a transaction takes effect somewhere between its invocation and response in 
	this execution.
For each transaction, we assign the real time at which the first correct server adds the transaction to its log and 
	we define an execution in which transactions are ordered according to this order and in which their responses
	are as they are for the executions of correct servers. 
This execution is consistent with the protocol order. It follows that it is equivalent to the canonical order and 
	therefore equivalent to the execution of every correct server.
Finally, the assigned time for a transaction is somewhere between its invocation and response because a response is	not produced before a correct server adds the transaction to its log.
\end{proof}

\subsection{Progress}

In periods of synchrony we assume there is an upper bound $d$ on
message transmission. For a correct client that sends messages to all
servers, we assume that the maximum delay in the reception of messages
between servers is $d$. We consider processing delay at correct
servers to be negligible.
\begin{lemma}
If $\delta > 2d$, all transactions of correct clients will be committed.
\end{lemma}
\begin{proof}
We consider a transaction \var{T} of a correct server that is initiated at real time $t_1$ and 
that is assigned timestamp $\hat{t}$. 
At time $t_1+3d$ every correct server will receive a message \var{Confirm}(\var{T},$\hat{t}$)
	for the transaction. 
The earliest time at which \var{T} can the $\mbox{\var{pending}}[t']$ set of a correct 
	server for some $t'$is $t_1+d$. 
So, the earliest time to propose to \var{CANCEL} \var{T} being 
	$t_1+d+\delta$ (Figure~\ref{fig:server}, Line~\ref{line:delay}). 
Since $\delta > 2d$, a correct server will receive \var{Confirm}(\var{T},$\hat{t}$)
	and adds \var{T} to \var{resolving} before time $t_1+\delta$ and the server will
	propose to \var{COMMIT} the transaction at that time.
So, we have shown that no correct process will propose to \var{CANCEL} a transaction by a correct client  in periods of synchrony.
It follows that the transaction will be committed.
In fact, all but the $t$ faulty servers will propose to \var{COMMIT} the transaction. 
By the validity requirement of the consensus protocol, the decision of all correct servers should be
	to \var{COMMIT}.
\end{proof}

\begin{lemma}[Pending transactions are resolved]\label{lemma:pending-resolved}
If $\mbox{\var{pending}}[t] \neq \emptyset$ at a correct server, all transactions in $\mbox{\var{pending}}[t] $ will eventually be resolved
	at every correct server.
\end{lemma}
\begin{proof}
Every transaction in $\mbox{\var{pending}}[t] \neq \emptyset$ will be proposed to be committed or cancelled (Figure~\ref{fig:server}, Line~\ref{line:COMMIT} or \ref{line:resolve:end})
	and every correct server will eventually commit or cancel the transaction.
\end{proof}

\begin{theorem}[Progress in periods of synchrony]
In periods of synchrony, all transactions of correct clients are applied.
\end{theorem}
\begin{proof}
  Lemma~\ref{lemma:pending-resolved} shows that all transactions of correct clients will 
	be committed in periods of synchrony. 
It remains to show that transactions of correct clients will be applied by the servers.
Consider a committed transaction  \var{T} of a correct client $c$ with timestamp $\hat{t}$. We need to show that 
	 \var{T} will be applied. 
At the time a correct server receives the \kConfirm{} message for \var{T}, it will calculate 
	the pending set if the set if empty and the set will have to contain the transaction 	
	\var{T}.
By Lemma~\ref{lemma:pending-resolved}, each transaction in that set will eventually be committed or cancelled.
This means that the termination condition (Figure~\ref{fig:apply},  Line~\ref{line:termination}) is guaranteed to be satisfied for
	every transaction that can be applied before \var{T}.
Each correct server will then be able to apply all non-cancelled transactions in \mbox{\var{confirmed}[t']} for every $t' < t$.
\end{proof}

\section{Byzantine Clients}\label{sec:byzantine-client}
We have assumed clients fail by crashing,
To tolerate Byzantine client behavior,
we must perform additional validation.
In particular, each client must prove that the the timestamp in its \kConfirm{}
message is correctly constructed by including the signed timestamps
Avoiding replay attacks is straightforward by having the servers sign
a cryptographic hash of the messages they send to the clients,
including transaction identifiers, which serve as nonces.

\section{Performance}\label{sec:perform}

To evaluate performance,
we adopt the definitions of \emph{gracious} and \emph{uncivil} executions from Clement et al.~\cite{Aardvark}.

\begin{definition}[Gracious execution~\cite{Aardvark}] An execution is \emph{gracious} if and only if (a) the execution is synchronous with some implementation-dependent short bound on message delay and (b) all clients and servers behave correctly.
\end{definition}

\begin{definition}[Uncivil execution~\cite{Aardvark}] An execution is \emph{uncivil} if and only if (a) the execution is synchronous with some implementation-dependent short bound on message delay, (b) up to $f$ servers and an arbitrary number of clients are Byzantine, and (c) all remaining clients and servers are correct.
\end{definition}

\subsection{Performance in gracious executions}

\subsubsection{Performance in the absence of contention}
In gracious executions,
and in the \emph{absence of contention},
the protocol requires 2.5 round-trip message delay from the time a client makes a request 
to the time it gets the result.  It takes one round-trip delay to receive the first response and calculate the timestamp $\hat{t}$.
It takes 1/2 round-trip delay for the servers to receive $\hat{t}$.
At that time, (1) correct servers initiate a consensus protocol to commit the transaction and (2) forward the \kConfirm{} message to all other servers. Another 1/2 round-trip delay later (2 round-trip delay from the start of the transaction), all correct servers decide to \var{COMMIT} the transaction (this is possible because all correct servers will be proposing the same \var{COMMIT} value) and calculate the pending sets.
At that point , the transaction can be executed (in the absence of  contention) and a response is received another 1/2 round-trip delay later at 2.5 round-trip delay from the start of the transaction

\subsubsection{Performance in the presence of contention}
In the presence of contention,
processing can be delayed by conflicting transactions that have the
same timestamp. The latest a transaction started after \var{T} can  be
assigned the same timestamp as \var{T} is just short of 1.5 round-trip
delay from the time \var{T} started (we assume that previous
transactions that are not concurrent with \var{T} have already been
cleared). In fact, a transaction that starts 1.5 round-trip delay
after \var{T} cannot  reach the servers before the time \var{T}'s
timestamp is propagated and will be assigned a later timestamp
(assuming the \kPropose{} message for the contending transaction
will propagate instantaneously in the worst case).  So, in the
presence of contention, a response might not arrive before 4
round-trip delays in gracious executions.  

It is important to note that these calculations are for one individual transaction delay and not the system throughput under load. The throughput is not affected by individual transaction latency.
The protocol is competitive in terms of latency with PBFT~\cite{PBFT} which achieves 2 round-trip delay in civil executions with a number of optimizations including speculative execution, but PBFT does not perform well in uncivil executions.


\subsection{Performance in uncivil executions}

In uncivil executions, the delay depends on the level of
contention. If a transaction is initiated and is not overlapping with
any other conflicting transaction, its delay will be the same as in
gracious executions. This is particularly good compared to
leader-based protocols in which a slow leader might need to be timed out before processing a transaction even if it does not conflict with other transactions.

In the presence of contention, a transaction can be delayed further. As in the gracious execution case, we consider the latest time a transaction can be added to the pending set of transaction \var{T}. As in the case of gracious executions, the time is 1.5 round-trip delay after \var{T} is initiated. If the client of the contending transaction fails, the full timeout would need to be incurred and a consensus protocol would need to be executed. So, the delay in this case would be  the timeout value $\delta$ plus the consensus time. The client will get a response by 0.5 round-trip delay after the consensus has ended (because the other message exchanges of the client overlap with the timeout time). 

\subsection{Other Performance Considerations}
It is  important to note that the delays are not additive. If we have
transactions with different timestamps,
and for each timestamp there is a slow pending transaction, no transaction incurs more than one timeout plus consensus delay because the timers are started in a pipelined fashion. This ensures that Byblos average throughput under client delays is minimally affected by slow clients. Also, recall that this delay is only incurred by conflicting transactions whereas in systems in which faulty servers are the source of the 
delay, all transactions are affected by server delays.

Another potential performance improvement that we did not consider is \emph{transaction batching}~\cite{PBFT} in which a number of transactions are processed in batches.  In our solution, servers communicate information about individual transactions. On the positive side, in Byblos, in the presence of contention, more transactions will get the same timestamp and the delay incurred for that timestamp is one for all transactions. This should improve throughput.

As described, Byblos uses public-key signatures~\cite{elgamal1985public,rivest1978method}, which can add significant overhead. 	Replacing signatures with message authentication codes~\cite{bellare1996keying} is a subject for future work.
Finally, the message complexity of our solution is rather high: $\mit{O}(n^2)$ messages per transaction.
Such high message complexity is not unusual for protocols that aim to achieve bounded delay (\cite{Prime,rBFT,Aardvark,BoundedByz} for example).

\section{Conclusion}\label{sec:conclude}
We showed that non-skipping timestamps can substantially simplify
state-machine replication,
facilitating a leaderless algorithm that exploits transaction
semantics to enhance concurrency.
For future work,
there are many ways to further improve Byblos' performance through
fast-path optimizations,
and through closely integrating a consensus protocol with the state-machine
replication solution.

\bibliographystyle{splncs03}
\bibliography{resolve}

\end{document}